\documentclass[a4paper, 12pt]{article}

\usepackage[sort&compress]{natbib}
\bibpunct{(}{)}{;}{a}{}{,} 

\usepackage{amsthm, amsmath, amssymb, mathrsfs, multirow, url, subfigure}
\usepackage{graphicx} 
\usepackage{ifthen} 
\usepackage{amsfonts}
\usepackage[usenames]{color}
\usepackage{fullpage}
\usepackage{setspace}

\newtheorem{theorem}{Theorem}

\newtheorem{proposition}{Proposition}
\newtheorem{lemma}{Lemma}
\newtheorem{assumption}{Assumption}

\newcommand{\PGamma}{P_\Gamma}
\newcommand{\PGammastar}{P_{\Gamma^\star}}
\newcommand{\fGamma}{f_\Gamma}
\newcommand{\fGammac}{f_{\Gamma^c}}
\newcommand{\fGammastar}{f_{\Gamma^\star}}
\newcommand{\fGammastarc}{f_{\Gamma^{\star c}}}
\newcommand{\FGamma}{F_\Gamma}
\newcommand{\FGammac}{F_{\Gamma^c}}
\newcommand{\FGammastar}{F_{\Gamma^\star}}
\newcommand{\FGammastarc}{F_{\Gamma^{\star c}}}

\newcommand{\RR}{\mathbb{R}}

\newcommand{\sdiff}{\triangle}

\newcommand{\veta}{\boldsymbol{\veta}}

\renewcommand{\phi}{\varphi} 
\newcommand{\eps}{\varepsilon}

\DeclareMathOperator{\diam}{diam}

\theoremstyle{plain}

\title{Robust and rate-optimal Gibbs posterior inference on the boundary of a noisy image}
\author{
Nicholas Syring\footnote{Department of Mathematics, Statistics, and Computer Science, University of Illinois at Chicago} \quad and \quad Ryan Martin\footnote{Department of Statistics, North Carolina State University, {\tt rgmarti3@ncsu.edu}}
}
\date{\today}

\begin{document}
\maketitle
\begin{abstract}
Detection of an image boundary when the pixel intensities are measured with noise is an important problem in image segmentation.
From a statistical point of view, the challenge is that likelihood-based methods require modeling the pixel intensities inside and outside the image boundary, even though these are typically of no practical interest.  Since misspecification of the pixel intensity models can negatively affect inference on the image boundary, it would be desirable to avoid this modeling step altogether.  Towards this, we develop a robust Gibbs approach that constructs a posterior distribution for the image boundary directly, without modeling the pixel intensities.  We prove that 
the Gibbs posterior concentrates asymptotically at the minimax optimal rate, adaptive to the boundary smoothness.  Monte Carlo computation of the Gibbs posterior is straightforward, and simulation results show that the corresponding inference is more accurate than that based on existing Bayesian methodology.
	
\smallskip


\end{abstract}

\section{Introduction}
\label{S:intro}

In image analysis, the boundary or edge of the image is one of the most important features of the image, and extraction of this boundary is a critical step.  An image consists of pixel locations and intensity values at each pixel, and the boundary can be thought of as a curve separating pixels of higher intensity from those of lower intensity.  Applications of boundary detection are wide-ranging, e.g., \citet{malik.2004} use boundary detection to identify important features in pictures of natural settings, \citet{Li.2010} identifies boundaries in medical images, and in \citet{yuan.2016} boundary detection helps classify the type and severity of wear on machines.  For images with noiseless intensity, boundary detection has received considerable attention in the applied mathematics and computer science literature; see, e.g., \citet{Ziou.Tabbone.1998}, \citet{Maini.Aggarwal.2009}, \citet{Li.2010}, and \citet{anam.2013}.  However, these approaches suffer from a number of difficulties.  First, they can produce an estimate of the image boundary, but do not quantify estimation uncertainty.  Second, these methods use a two-stage approach where the image is first smoothed to filter out noise and then a boundary is estimated based on a calculated intensity gradient.  This two-stage approach makes theoretical analysis of the method challenging, and no convergence results are presently known to the authors.  Third, in our examples, these methods perform poorly on noisy data, and we suspect one reason for this is that the intensity gradient is less informative for the boundary when we observe noisy data.  In the statistics literature, \citet{gu.etal.2015} take a Bayesian approach to boundary detection and emphasize borrowing information to recover boundaries of multiple, similar objects in an image.  Boundary detection using wombling is also a popular approach; see \citet{Liang.etal.2009}, with applications to geography \citep{Lu.Carlin.2004}, public health \citep{Ma.Carlin.2007}, and ecology \citep{Fitz.etal.2010}.  However, these techniques are used with areal or spatially aggregated data and are not suitable for the pixel data encountered in image analysis. 

In Section~\ref{S:setup}, we give a detailed description of the image boundary problem, following the setup in \citet{Li.Ghosal.2015}.  They take a fully Bayesian approach, modeling the probability distributions of the pixel intensities both inside and outside the image.  This approach is challenging because it often introduces nuisance parameters in addition to the image boundary.  Therefore, in Section~\ref{S:main}, we propose to use a so-called \emph{Gibbs model}, where a suitable loss function is used to connect the data to the image boundary directly, rather than a probability model \citep[e.g.,][]{Catonib, zhang.2006, jiang.tanner.2008, syring.martin.mcid, walker.2016}. 

We investigate the asymptotic convergence properties of the Gibbs posterior, in Section~\ref{S:post_concentration}, and we show that, if the boundary is $\alpha$-H\"older smooth, then the Gibbs posterior concentrates around the true boundary at the rate $\{(\log n)/n\}^{\alpha / (\alpha + 1)}$, which is minimax optimal \citep{mammen.1995} up to the logarithmic factor, relative to neighborhoods of the true boundary measured by the Lebesgue measure of a symmetric difference.  Moreover, as a consequence of appropriately mixing over the number of knots in the prior, this rate is adaptive to the unknown smoothness $\alpha$.  To our knowledge, this is the first Gibbs posterior convergence rate result for an infinite-dimensional parameter, so the proof techniques used herein may be of general interest.  Further, since the Gibbs posterior concentrates at the optimal rate without requiring a model for the pixel intensities, we claim that the inference on the image boundary is robust.  

Computation of the Gibbs posterior is relatively straightforward and, in Section~\ref{S:computation}, we present a reversible jump Markov chain Monte Carlo method; $\mathtt{R}$ code to implement to the proposed Gibbs posterior inference is available at \url{https://github.com/nasyring/GibbsImage}.  A comparison of inference based on the proposed Gibbs posterior to that based on the fully Bayes approach in \citet{Li.Ghosal.2015} is shown in Section~\ref{S:Simulations}.  For smooth boundaries, the two methods perform similarly, providing very accurate estimation.  However, the Gibbs posterior is easier to compute, thanks to there being no nuisance parameters, and is notably more accurate than the Bayes approach when the model is misspecified or when the boundary is not everywhere smooth.
The technical details of the proofs are deferred to Appendix~\ref{S:proofs}.

\section{Problem formulation}
\label{S:setup}

Let $\Omega \subset \RR^2$ be a bounded region that represents the frame of the image; typically, $\Omega$ will be a square, say, $[-\frac12, \frac12]^2$, but, generally, we assume only that $\Omega$ is scaled to have unit Lebesgue measure.  Data consists of pairs $(X_i, Y_i)$, $i=1,\ldots,n$, where $X_i$ is a pixel location in $\Omega$ and $Y_i$ is an intensity measurement at that pixel.  The range of $Y_i$ is context-dependent, and we consider both binary and real-valued cases below.  The model asserts that there is a region $\Gamma \subset \Omega$ such that the intensity distribution is different depending on whether the pixel is inside or outside $\Gamma$.  We consider the following model for the joint distribution $\PGamma$ of pixel location and intensity, $(X,Y)$: 
\begin{align}
X & \sim g(x), \notag \\
Y \mid (X=x) & \sim \fGamma(y) \, 1(x \in \Gamma) + \fGammac(y) \, 1(x \in \Gamma^c), \label{eq:gamma}
\end{align}
where $g$ is a density on $\Omega$, $\fGamma$ and $\fGammac$ are densities on the intensity space, $\FGamma$ and $\FGammac$ are their respective distribution functions, and $1(\cdot)$ denotes an indicator function.  That is, given the pixel location $X=x$, the distribution of the pixel intensity $Y$ depends only on whether $x$ is in $\Gamma$ or $\Gamma^c$.  Of course, more complex models are possible, e.g., where the pixel intensity distribution depends on the specific pixel location, but we will not consider such models here.  We assume that there is a true, star-shaped set of pixels, denoted by $\Gamma^\star$, with a known reference point in its interior.  That is, any point in $\Gamma^\star$ can be connected to the reference point by a line segment fully contained in $\Gamma^\star$.  The observations $\{(X_i, Y_i): i=1,\ldots,n\}$ are iid samples from $\PGammastar$, and the goal is to make inference on $\Gamma^\star$ or, equivalently, its boundary $\gamma^\star = \partial \Gamma^\star$.  

The density $g$ for the pixel locations is of no interest and is taken to be known.  The question is how to handle the two conditional distributions, $\fGamma$ and $\fGammac$.  \citet{Li.Ghosal.2015} take a fully Bayesian approach, modeling both $\fGamma$ and $\fGammac$.  By specifying these parametric models, they are obligated to introduce priors and carry out posterior computations for the corresponding parameters.  Besides the efforts needed to specify models and priors and to carry out posterior computations, there is also a concern that the pixel intensity models might be misspecified, potentially biasing the inference on $\Gamma^\star$.  Since the forms of $\fGamma$ and $\fGammac$, as well as any associated parameters, are irrelevant to the boundary detection problem, it is natural to ask if inference can be carried out robustly, without modeling the pixel intensities.  

We answer this question affirmatively, developing a Gibbs model for $\Gamma$ in Section~\ref{S:main}.  In the present context, suppose we have a loss function $\ell_\Gamma(x,y)$ that measures how well an observed pixel location--intensity pair $(x,y)$ agrees with a particular region $\Gamma$.  The defining characteristic of $\ell_\Gamma$ is that $\Gamma^\star$ should be the unique minimizer of $\Gamma \mapsto R(\Gamma)$, where $R(\Gamma) = \PGammastar \ell_\Gamma$ is the risk, i.e., the expected loss.  A main contribution here, in Section~\ref{SS:loss}, is specification of a loss function that meets this criterion.  A necessary condition to construct such a loss function is that the distribution functions $\FGamma$ and $\FGammac$ are stochastically ordered.  Imagine a gray-scale image; then, stochastic ordering means this image is lighter, on average, inside the boundary than outside the boundary, or vice-versa.  In the specific context of binary images the pixel densities $\fGammastar$ and $\fGammastarc$ are simply numbers between 0 and 1, and the stochastic ordering assumption means that, without loss of generality, $\fGammastar > \fGammastarc$, while for continuous images, again without loss of generality, $\FGammastar(y) < \FGammastarc(y)$ for all $y \in \RR$.  If we define the empirical risk, 
\begin{equation}
\label{eq:empirical.risk}
R_n(\Gamma) = \frac1n \sum_{i=1}^n \ell_\Gamma(X_i, Y_i), 
\end{equation}
then, given a prior distribution $\Pi$ for $\Gamma$, 
the Gibbs posterior distribution for $\Gamma$, denoted by $\Pi_n$, has the formula
\begin{equation}
\label{eq:gibbs}
\Pi_n(B) = \frac{\int_B e^{-n R_n(\Gamma)} \, \Pi(d\Gamma)}{\int e^{-n R_n(\Gamma)} \, \Pi(d\Gamma)}, 
\end{equation}
where $B$ is a generic $\Pi$-measurable set of $\Gamma$'s.  Of course, \eqref{eq:gibbs} only makes sense if the denominator is finite; in Section~\ref{S:main}, our risk functions $R_n$ are non-negative so this integrability condition holds automatically.  

Proper scaling of the loss in the Gibbs model is important \citep[e.g.,][]{walker.2016, syring.martin.gibbs}, and here we will provide some context-specific scaling; see Sections~\ref{S:post_concentration} and \ref{SS:scaling}.  Our choice of prior $\Pi$ for $\Gamma$ is discussed in Section~\ref{SS:prior}.  
Together, the loss function $\ell_\Gamma$ and the prior for $\Gamma$ define the Gibbs model, no further modeling is required.

\section{Gibbs model for the image boundary}
\label{S:main}

\subsection{Loss function}
\label{SS:loss}

To start, we consider inference on the image boundary when the pixel intensity is binary, i.e., $Y_i \in \{-1,+1\}$.  In this case, the densities, $\fGamma$ and $\fGammac$, in \eqref{eq:gamma} must be Bernoulli, so the likelihood is known.  Even though the parametric form of the conditional distributions is known, the Gibbs approach only requires prior specification and posterior computation related to $\Gamma$, whereas the Bayes approach must also deal with the nuisance parameters in these Bernoulli conditional distributions.  The binary case is relatively simple and will provide insights into how to formulate a Gibbs model in the more challenging continuous intensity problem.   

A reasonable choice for the loss function $\ell_\Gamma$ is the following weighted misclassification error loss, depending on a parameter $h>0$:
\begin{equation}
\label{eq:mce_w}
\ell_\Gamma(x,y) = \ell_\Gamma(x,y \mid h) = h \, 1(y=+1, x \in \Gamma^c) + 1(y=-1, x \in \Gamma). 
\end{equation}
Note that putting $h=1$ in \eqref{eq:mce_w} gives the usual misclassification error loss.  In order for the Gibbs model to work the risk, or expected loss, must be minimized at the true $\Gamma^\star$ for some $h$.  Picking $h$ to ensure this property holds necessitates making a connection between the probability model in \eqref{eq:gamma} and the loss in \eqref{eq:mce_w}.  The condition in \eqref{eq:h.bound} below is just the connection needed.  With a slight abuse of notation, let $\fGammastar$ and $\fGammastarc$ denote the conditional probabilities for the event $Y=+1$, given $X \in \Gamma^\star$ and $X \in \Gamma^{\star c}$, respectively.  Recall our stochastic ordering assumption implies $\fGammastar > \fGammastarc$.  

\begin{proposition}
\label{prop:min.binary}
Using the notation in the previous paragraph, if $h$ is such that 
\begin{equation}
\label{eq:h.bound}
\fGammastar > \frac{1}{1+h} > \fGammastarc, 
\end{equation}
then the risk function $R(\Gamma) = \PGammastar \ell_\Gamma$ is minimized at $\Gamma^\star$.
\end{proposition}


Either one---but not both---of the above inequalities can be made inclusive and the result still holds.  The condition in \eqref{eq:h.bound} deserves additional explanation.  For example, if we know $\fGammastar \geq \frac{1}{2} > \fGammastarc$, then we take $h=1$, which means that in \eqref{eq:mce_w} we penalize both intensities of $1$ outside $\Gamma$ and intensities of $-1$ inside $\Gamma$ by a loss of $1$.  If, however, we know the overall image brightness is higher so that $\fGammastar \geq \frac{4}{5} > \fGammastarc$ then we take $h = 1/4$ in \eqref{eq:mce_w} and penalize bright pixels outside $\Gamma$ by less than dull pixels inside $\Gamma$.  To see why this loss balancing is so crucial, suppose the second case above holds so that $\fGammastar=4/5$ and $\fGammastarc = 3/4$, but we take $h = 1$ anyway.  Then, in \eqref{eq:mce_w}, $1(y=+1, x \in \Gamma^c)$ is very often equal to $1$ while $1(y=-1, x \in \Gamma)$ is very often $0$.  We will likely minimize the expected loss then by incorrectly taking $\Gamma$ to be all of $\Omega$ so that the first term in the loss vanishes.  Knowing a working $h$ corresponds to having some prior information about $\fGammastar$ and $\fGammastarc$, but we can also use the data to estimate a good value of $h$ and we describe this data-driven strategy in Section~\ref{SS:scaling}.
  
Next, for the continuous case, we assume that the pixel intensity takes value in $\RR$.  The proposed strategy is to modify the misclassification error  \eqref{eq:mce_w} by working with a suitably discretized pixel intensity measurement, reminiscent of threshold modeling.  In particular, consider the following version of the misclassification error, depending on parameters $(c, k, z)$, with $c,k>0$:
\begin{equation}
\label{eq:mce_cont_w}
\ell_\Gamma(x,y) = \ell_\Gamma(x,y \mid c, k, z) = k \, 1(y > z, x \in \Gamma^c) + c \, 1(y \leq z, x \in \Gamma). 
\end{equation}
Again, we claim that, for suitable $(c, k, z)$, the risk function is minimized at $\Gamma^\star$.  Let $\FGamma$ and $\FGammac$ denote the distribution functions corresponding to the densities $\fGamma$ and $\fGammac$ in \eqref{eq:gamma}, respectively.  Recall our stochastic ordering assumption implies $\FGammastar(z) < \FGammastarc(z)$.  

\begin{proposition}
\label{prop:min.continuous}
If $(c, k, z)$ in \eqref{eq:mce_cont_w} satisfies 
\begin{equation}
\label{eq:ckz.bound.1}
\FGammastar(z) < \frac{k}{k+c} < \FGammastarc(z),
\end{equation}
then the risk function $R(\Gamma) = \PGammastar \ell_\Gamma$ is minimized at $\Gamma^\star$.
\end{proposition}


Again, either one---but not both---of the above inequalities in \eqref{eq:ckz.bound.1} can be made inclusive and the result still holds.  The parameters $k$ and $c$ in \eqref{eq:mce_cont_w} determine the scale of the loss as mentioned in Section~\ref{S:intro}, while $z$ determines an intensity cutoff.  According to the loss in \eqref{eq:mce_cont_w}, if a given pixel is located at $x\in \Gamma^c$, and with intensity $y$ larger than cutoff $z$, it will incur a loss of $k>0$.  This implies that the true image $\Gamma^\star$ can be identified by working with a suitable version of the loss \eqref{eq:mce_cont_w}.  A similar condition to \eqref{eq:ckz.bound.1}, see Assumption~\ref{A:ckz} in Section~\ref{S:post_concentration}, says what scaling is needed in order for the Gibbs posterior to concentrate at the optimal rate.  Although the conditions on the scaling all involve the unknown distribution $\PGammastar$, a good choice of $(c, k, z)$ can be made based on the data alone, without prior information, and we discuss this strategy in Section~\ref{SS:scaling}.

\subsection{Prior specification}
\label{SS:prior}

We specify a prior distribution for the boundary of the region $\Gamma$ by first expressing the pixel locations $x$ in terms of polar coordinates $(\theta,r)$, an angle and radius, where $\theta \in [0,2\pi]$ and $r > 0$.  The specific reference point and angle in $\Omega$ used to define polar coordinates are essentially arbitrary, subject to the requirement that any point in $\Gamma^\star$ can be connected to the reference point by a line segment contained in $\Gamma^\star$.  \citet{Li.Ghosal.2015} tested the influence of the reference point in simulations and found it to have little influence on the results.  Using polar coordinates the boundary of $\Gamma$ can be determined by the parametric curve $(\theta, \gamma(\theta))$.  We proceed to model this curve $\gamma$.  

Whether one is taking a Bayes or Gibbs approach, a natural strategy to model the image boundary is to express $\gamma$ as a linear combination of suitable basis functions, i.e., $\gamma(\theta) = \hat{\gamma}_{D, \beta}(\theta) = \sum_{j = 1}^{D} \beta_j B_{j,D}(\theta)$.  \citet{Li.Ghosal.2015} use the eigenfunctions of the squared exponential periodic kernel as their basis functions.  Here we consider a model based on free knot b-splines, where the basis functions are defined recursively as 
\begin{align*}
B_{i,1}(\theta) & = 1(\theta \in [t_i, t_{i+1}]) \\
B_{i,k}(\theta) & = \frac{\theta - t_i}{t_{i+k-1} - t_i}B_{i,k-1}(\theta) + \frac{t_{i+k}-\theta}{t_{i+k} - t_{i+1}}B_{i+1, k-1}(\theta),
\end{align*}
where $\theta \in [0,2\pi]$, $t_{-2},t_{-1},t_{0}$ and $t_{D+1}, t_{D+2}, t_{D+3}$ are outer knots, $t_{1}, ..., t_{D}$ are inner knots, and $\beta = (\beta_1, ..., \beta_{D}) \in (\RR^+)^D$ is the vector of coefficients.  Note that we restrict the coefficient vector $\beta$ to be positive because the function values $\gamma(\theta)$ measure the radius of a curve from the origin.  In the simulations in Section~\ref{S:Simulations}, the coefficients $\beta_2, ..., \beta_D$ are free parameters, while $\beta_1$ is calculated deterministically to force the boundary to be closed, i.e. $\gamma(0) = \gamma(2\pi)$, and we require $t_{1} = 0$ and $t_{D} = 2\pi$; all other inner knots are free.  Our model based on the b-spline representation seem to perform as well as the eigenfunctions used in \citet{Li.Ghosal.2015} for smooth boundaries, but a bit better for boundaries with corners; see Section~\ref{S:Simulations}.  

Therefore, the boundary curve is $\gamma$ is parametrized by an integer $D$ and a $D$-vector $\beta$.  We introduce a prior $\Pi$ on $(D,\beta)$ hierarchically as follows: $D$ has a Poisson distribution with rate $\mu_D$ and, given $D$, the coordinates $\beta_1,\ldots,\beta_D$ of $\beta$ are iid exponential with rate $\mu_\beta$.  These choices satisfy the technical conditions on $\Pi$ detailed in Section~\ref{S:post_concentration}.  In our numerical experiments in Section~\ref{S:Simulations}, we take $\mu_D=12$ and $\mu_\beta = 10$.

\section{Gibbs posterior convergence}
\label{S:post_concentration}

The Gibbs model depends on two inputs, namely, the prior and the loss function.  In order to ensure that the Gibbs posterior enjoys desirable asymptotic properties, some conditions on both of these inputs are required.  The first assumption listed below concerns the loss; the second concerns the true image boundary $\gamma^\star = \partial \Gamma^\star$; and the third concerns the prior.  Here we will focus on the continuous intensity case, since the only difference between this and the binary case is that the latter provides the discretization for us.  

\begin{assumption}
\label{A:ckz}
Loss function parameters $(c, k, z)$ in \eqref{eq:mce_cont_w} satisfy 
\begin{equation}
\label{eq:ckz.bound.2}
\FGammastar(z) < \frac{e^k-1}{e^{c+k}-1} \quad \text{and} \quad \FGammastarc(z) > \frac{e^k - 1}{e^k - e^{-c}}. 
\end{equation}
\end{assumption}  

Compared to the condition \eqref{eq:ckz.bound.1} that was enough to allow the loss function to identify the true $\Gamma^\star$, condition \eqref{eq:ckz.bound.2} in Assumption~\ref{A:ckz} is only slightly stronger.  This can be seen from the following inequality:
\[ \frac{e^k - 1}{e^k - e^{-c}} > \frac{k}{k+c} > \frac{e^k-1}{e^{c+k}-1}. \]
However, if $(c,k)$ are small, then the three quantities in the above display are all approximately equal, so Assumption~\ref{A:ckz} is not much stronger than what is needed to identify $\Gamma^\star$.  Again, these conditions on $(c, k, z)$ can be understood as providing a meaningful scale to the loss function.  Intuitively, the scale of the loss between observations receiving no loss versus some loss, expressed by parameters $k$ and $c$, should be related to the level of information in the data.  When $\FGammastar(z)$ and $\FGammastarc(z)$ are far apart, the data can more easily distinguish between $\FGammastar$ and $\FGammastarc$, so we are free to assign larger losses than when $\FGammastar(z)$ and $\FGammastarc(z)$ are close and the data are relatively less informative.

The ability of a statistical method to make inference on the image boundary will depend on how smooth the true boundary is.  \citet{Li.Ghosal.2015} interpret $\gamma^\star$ as a function from the unit circle to the positive real line, and they formulate a H\"older smoothness condition for this function.  Following the prior specification described in Section~\ref{SS:prior}, we treat the boundary $\gamma^\star$ as a function from the interval $[0,2\pi]$ to the positive reals, and formulate the smoothness condition on this arguably simpler version of the function.  Since the reparametrization of the unit circle in terms of polar coordinates is smooth, it is easy to check that the H\"older smoothness condition \eqref{eq:holder} below is equivalent to that in \citet{Li.Ghosal.2015}.  
 
\begin{assumption}
\label{A:true}
The true boundary function $\gamma^\star: [0,2\pi] \to \RR^+$ is $\alpha$-H\"older smooth, i.e., there exists a constant $L=L_{\gamma^\star} > 0$ such that 
\begin{equation}
\label{eq:holder}
|(\gamma^\star)^{([\alpha])}(\theta) - (\gamma^\star)^{([\alpha])}(\theta')| \leq L |\theta - \theta'|^{\alpha - [\alpha]}, \quad \forall \; \theta,\theta' \in [0,2\pi], 
\end{equation}
where $(\gamma^\star)^{(k)}$ denotes the $k^\text{th}$ derivative of $\gamma^\star$ and $[\alpha]$ denotes the largest integer less than or equal to $\alpha$.  Following the description of $\Gamma^\star$ in the introduction, we also assume that the reference point is strictly interior to $\Gamma^\star$ meaning that it is contained in an open set itself wholly contained in $\Gamma^\star$ so that $\gamma^\star$ is uniformly bounded away from zero.  Moreover, the density $g$ for $X$, as in \eqref{eq:gamma}, is uniformly bounded above by $\overline{g} := \sup_{x \in \Omega} g(x)$ and below by $\underline{g} := \inf_{x \in \Omega} g(x) \in (0,1)$ on $\Omega$.  
\end{assumption}

General results are available on the error in approximating an $\alpha$-H\"older smooth function by b-splines of the form specified in Section~\ref{SS:prior}.  Indeed, Theorem~6.10 in \citet{Schumaker.2007} implies that if $\gamma^\star$ satisfies \eqref{eq:holder}, then 
\begin{equation} 
\label{eq:bsplineapprox}
\text{$\forall$ $d > 0$, $\exists$ $\beta_d^\star \in (\RR^+)^d$ such that $\|\gamma^\star - \hat\gamma_{d,\beta_d^\star}\|_\infty \lesssim d^{-\alpha}$.}
\end{equation}
Since $\gamma^\star(\theta)>0$, we can consider all coefficients to be positive; i.e. $\beta_d^\star \in (\RR^+)^d$ and see Lemma~1(b) in \citet{shen.ghosal.2015}.  The next assumption about the prior makes use of the approximation property in \eqref{eq:bsplineapprox}.  


\begin{assumption}
\label{A:prior}
Let $\beta_d^\star$, for $d > 0$, be as in \eqref{eq:bsplineapprox}.  Then there exists $C, m > 0$ such that the prior $\Pi$ for $(D,\beta)$ satisfies, for all $d > 1$,  
\begin{align*}
\log \Pi(D > d) & \lesssim -d \log d, \\
\log \Pi(D=d) & \gtrsim -d \log d, \\
\log \Pi(\|\beta - \beta_d^\star\|_1 \leq kd^{-\alpha} \mid D=d) & \gtrsim -d\log\{1/(kd^{-\alpha})\}, \\
\log \Pi(\beta \not\in [-m,m]^d \mid D=d) & \lesssim \log d - C m.
\end{align*}
\end{assumption}

The first two conditions in Assumption~\ref{A:prior} ensure that the prior on $D$ is sufficiently spread out while the second two conditions ensure that there is sufficient prior support near $\beta$'s that approximate $\gamma^\star$ well.  Assumption~\ref{A:prior} is also needed in \citet{Li.Ghosal.2015} for convergence of the Bayesian posterior at the optimal rate.  However, the Bayes model also requires assumptions about the priors on the nuisance parameters, e.g., Assumption~C in \citet{Li.Ghosal.2015}, which are not necessary in our approach here.  

In what follows, let $A \sdiff B$ denote the symmetric difference of sets in $\Omega$ and $\lambda(A \sdiff B)$ its Lebesgue measure. 


\begin{theorem}
\label{thm:rate}
With a slight abuse of notation, let $\Pi$ denote the prior for $\Gamma$, induced by that on $(D,\beta)$, and $\Pi_n$ the corresponding Gibbs posterior \eqref{eq:gibbs}.  Under Assumptions~\ref{A:ckz}--\ref{A:prior}, there exists a constant $M > 0$ such that 
\[ \PGammastar \Pi_n(\{\Gamma: \lambda(\Gamma^\star \sdiff \Gamma) > M\eps_n\}) \to 0 \quad \text{as $n \to \infty$}, \]
where $\eps_n = \{(\log n)/n\}^{\alpha/(\alpha+1)}$ and $\alpha$ is the smoothness coefficient in Assumption~\ref{A:true}.   
\end{theorem}


Theorem~\ref{thm:rate} says that, as the sample size increases, the Gibbs posterior places its mass on a shrinking neighborhood of the true boundary $\gamma^\star$.  The rate, given by the size of the neighborhood, is optimal according to \citet{mammen.1995}, up to a logarithmic factor, and adaptive since the prior does not depend on the unknown smoothness $\alpha$.  

\section{Computation}
\label{S:computation}

\subsection{Sampling algorithm}
\label{SS:rjmcmc}

We use reversible jump Markov chain Monte Carlo, as in \citet{green.1995}, to sample from the Gibbs posterior.  These methods have been used successfully in Bayesian free-knot spline regression problems; see, e.g., \citet{denison.1998} and \citet{dimatteo.2001}.  Although the sampling procedure is more complicated when allowing the number and locations of knots to be random versus using fixed knots, the resulting spline functions can do a better job fitting curves with low smoothness.  

To summarize the algorithm, we start with the prior distribution $\Pi$ for $(D,\beta)$ as discussed in Section~\ref{SS:prior}.  Next, we need to initialize values of $D$, the knot locations $\{t_{-2}, ..., t_{D+3}\}$, and the values of $\beta_2, ..., \beta_D$.  The value of $\beta_1$ is then calculated numerically to force closure.  We choose $D = 12$ with $t_{-2} = -2$, $t_{-1} = -1$, $t_{0} = -0.5$, $t_{13} = 2\pi+0.5$, $t_{14} = 2\pi+1$, $t_{15} = 2\pi+2$ and $t_{1}, ..., t_{12}$ evenly spaced in $[0,2\pi]$.  We set inner knots $t_{0} = 0$ and $t_{D} = 2\pi$ while the other inner knot locations remain free to change in birth, death, and relocation moves; we also set $\beta_2 = \beta_3 = ...= \beta_{12} = 0.1$.  Then the following three steps constitutes a single iteration of the reversible jump Markov chain Monte Carlo algorithm to be repeated until the desired number of samples are obtained:
\begin{enumerate}
\item Use Metropolis-within-Gibbs steps to update the elements of the $\beta$ vector, again solving for $\beta_1$ to force closure at the end.  In our examples we use normal proposals centered at the current value of the element of the $\beta$ vector, and with standard deviation $0.10$.
\item Randomly choose to attempt either a birth, death, or relocation move to add a new inner knot, delete an existing inner knot, or move an inner knot.  
\item Attempt the jump move proposed in Step~2.  The $\beta$ vector must be appropriately modified when adding or deleting a knot, and again we must solve for $\beta_1$.  Details on the calculation of acceptance probabilities for each move can be found in \citet{denison.1998} and \citet{dimatteo.2001}.  
\end{enumerate}
R code to implement this Gibbs posterior sampling scheme, along with the empirical loss scaling method described in Section~\ref{SS:scaling}, is available at \url{https://github.com/nasyring/GibbsImage}.

\subsection{Loss scaling}
\label{SS:scaling}

It is not clear how to select $(c, k, z)$ to satisfy Assumption~\ref{A:ckz} without knowledge of $\FGammastar$ and $\FGammastarc$.  However, it is fairly straightforward to select values of $(c, k, z)$ based on the data which are likely to meet the required condition.  First, we need a notion of optimal $(c, k, z)$ values.  If we knew $\FGammastar$ and $\FGammastarc$, then we would select $z$ to maximize $\FGammastarc(z) - \FGammastar(z)$ because this choice of $z$ gives us the point at which $\FGammastar$ and $\FGammastarc$ are most easily distinguished.  Then, we would choose $(c, k)$ to be the largest values such that \eqref{eq:ckz.bound.2} holds.  Intuitively, we want large values of $(c, k)$ so that the loss function in \eqref{eq:mce_cont_w} is more sensitive to departures from $\gamma^\star$.   

Since we do not know $\FGammastar$ and $\FGammastarc$, we estimate $\FGammastar(z)$ and $\FGammastarc(z)$ from the data.  In order to do this, we need a rough estimate of $\gamma^\star$ to define the regions $\Gamma^\star$ and $\Gamma^{\star c}$.  Our approach is to model $\gamma$ with a b-spline, as before, and estimate $\gamma^\star$ several times by minimizing \eqref{eq:mce_cont_w} using several different values of the loss scaling parameters $(c,k,z)$.  Specifically, set a grid of $z$ values $z_1, z_2, ...,z_g$, and for each $z_j$, find $(c, k) = (c_j, k_j)$ that satisfy 
\[ \frac{k_j}{k_j+c_j} = \frac{|\{i: y_i \leq z_j\}|}{n}. \]
Next, estimate $\gamma^\star$ by minimizing \eqref{eq:mce_cont_w} using $(c_j, k_j, z_j)$.  The estimate of $\gamma^\star$ provides estimates of the regions $\Gamma^\star$ and $\Gamma^{\star c}$ which we use to calculate the sample proportions 
\[\hat F_{\Gamma^\star} := \frac{|\{i: y_i \leq z_j, x_i \in \hat{\Gamma}^\star\}|}{|\{i:x_i \in \hat{\Gamma}^\star\}|}\, \hspace{3mm}{\rm{  and  }}\hspace{3mm}\,  \hat F_{\Gamma^{\star c}}:= \frac{|\{i: y_i \leq z_j, x_i \in \hat{\Gamma}^{\star c}\}|}{|\{i:x_i \in \hat{\Gamma}^{\star c}\}|}. \]
Then, the approximately optimal value is $z = \arg\max_{z_j} \hat F_{\Gamma^{\star c}}(z_j) - \hat F_{\Gamma^\star}(z_j)$.  Finally, choose the approximately optimal values of $(c, k)$ to satisfy \eqref{eq:ckz.bound.2} replacing $\FGammastar(z)$ and $\FGammastarc(z)$ by their estimates $\hat F_{\Gamma^\star}(z)$ and $\hat F_{\Gamma^{\star c}}(z)$.

Based on the simulations in Section \ref{S:Simulations}, this method produces values of $(c, k, z)$ very close to their optimal values.  Importantly, the estimated $(c, k)$ are more likely to be smaller than their optimal values than larger, which makes our estimates more likely to satisfy \eqref{eq:ckz.bound.2}.  This is a consequence of the stochastic ordering of $\FGammastar$ and $\FGammastarc$.  Unless the classifier we obtain by minimizing \eqref{eq:mce_cont_w} is perfectly accurate, we will tend to mix together samples from $\FGammastar$ and $\FGammastarc$ in our estimates.  If we estimate $\FGammastar(z)$ with some observations from $\FGammastar$ and some from $\FGammastarc$, we will tend to overestimate $\FGammastar(z)$, and vice versa we will tend to underestimate $\FGammastarc(z)$.  These errors will cause $(c, k)$ to be underestimated, and therefore more likely to satisfy \eqref{eq:ckz.bound.2}.

\section{Numerical examples}
\label{S:Simulations}
We tested our Gibbs model on data from both binary and continuous images following much the same setup as in \citet{Li.Ghosal.2015}.  The pixel locations in $\Omega = [-\frac12,\frac12]^2$ are sampled by starting with a fixed $m \times m$ grid in $\Omega$ and making a small random uniform perturbation at each grid point.  Several different pixel intensity distributions are considered.  We consider two types of shapes for $\Gamma^\star$: an ellipse with center $(0.1, 0.1)$, rotated at an angle of 60~degrees, with major axis length $0.35$ and minor axis length $0.25$; and a centered equilateral triangle of height 0.5.  The ellipse boundary will test the sensitivity of the model to boundaries which are off-center while the triangle tests the model's ability to identify non-smooth boundaries.  

We consider four binary and four continuous intensity images and compare with the Bayesian method of \citet{Li.Ghosal.2015} as implemented in the {\em BayesBD} package \citep{BayesBD} available on CRAN .

\begin{itemize}
\item[B1.]  Ellipse image, $m=100$, and $\FGammastar$ and $\FGammastarc$ are Bernoulli with parameters 0.5 and 0.2, respectively.
\item[B2.]  Same as B1 but with triangle image.    
\item[B3.]  Ellipse image, $m=500$, and $\FGammastar$ and $\FGammastarc$ are Bernoulli with parameters 0.25 and 0.2, respectively.  
\item[B4.]  Same as B3 but with triangle image.   

\item[C1.]  Ellipse image, $m=100$, and $\FGammastar$ and $\FGammastarc$ are $N(4, 1.5^2)$ and $N(1, 1)$, respectively.   
\item[C2.]  Same as C1 but with triangle image.     
\item[C3.]  Ellipse image, $m=100$, and $\FGammastar$ and $\FGammastarc$ are $0.2\, N(2, 10) + 0.8 \,N(0, 1)$, a normal mixture, and $N(0, 5)$, respectively. 
\item[C4.]  Ellipse image, $m=100$, and $\FGammastar$ and $\FGammastarc$ are $t$ distributions with 3 degrees of freedom and non-centrality parameters 1 and 0, respectively. 
\end{itemize}

For binary images, the likelihood must be Bernoulli, so the Bayesian model is correctly specified in cases B1--B4.  For the continuous examples in C1--C4, we assume a Gaussian likelihood for the Bayesian model.  Then, cases C1 and C2 will show whether or not the Gibbs model can compete with the Bayesian model when the model is correctly specified, while cases C3 and C4 will demonstrate the superiority of the Gibbs model over the Bayesian model under misspecification.  Again, the Gibbs model has the added advantage of not having to specify priors for or sample values of the mean and variance  associated with the normal conditional distributions.   	

We replicated each example scenario 100 times for both the Gibbs and Bayesian models, each time producing a posterior sample of size 4000 after a burn in of 1000 samples.  We recorded the errors---Lebesgue measure of the symmetric difference---for each run along with the estimated loss function parameters for the Gibbs models for continuous images.  The results are summarized in Tables~\ref{table:errors}--\ref{table:loss}.  We see that the Gibbs model is competitive with the fully Bayesian model in Examples B1--B4, C1, and C2, when the likelihood is correctly specified.  When the likelihood is misspecified, there is a chance that the Bayesian model will fail, as in Examples C3 and C4.  However, the Gibbs model does not depend on a likelihood, only the stochastic ordering of $\FGammastar$ and $\FGammastarc$, and it continues to perform well in these non-Gaussian examples.  From Table~\ref{table:loss}, we see that the empirical method described in Section~\ref{SS:scaling} is able to select parameters for the loss function in \eqref{eq:mce_cont_w} close to the optimal values and meeting Assumption~\ref{A:ckz}.

\begin{table}
\caption{Average errors (and standard deviations) for each example.}
\label{table:errors}
\centering
\begin{tabular}{lcccccccc}

																																							\hline
Model & B1                                                & B2                                                 & B3                                                 & B4                                                 & C1                                                 & C2                                                 & C3                                                 & C4                                                 \\ \hline
Bayes   & \begin{tabular}[c]{@{}c@{}}0.00\\ (0.00)\end{tabular} & \begin{tabular}[c]{@{}c@{}}0.02\\ (0.00)\end{tabular} & \begin{tabular}[c]{@{}c@{}}0.01\\ (0.00)\end{tabular} & \begin{tabular}[c]{@{}c@{}}0.02\\ (0.01)\end{tabular} & \begin{tabular}[c]{@{}c@{}}0.03\\ (0.03)\end{tabular} & \begin{tabular}[c]{@{}c@{}}0.04\\ (0.03)\end{tabular} & \begin{tabular}[c]{@{}c@{}}0.11\\ (0.06)\end{tabular} & \begin{tabular}[c]{@{}c@{}}0.10\\ (0.05)\end{tabular} \\
Gibbs   & \begin{tabular}[c]{@{}c@{}}0.01\\ (0.00)\end{tabular} & \begin{tabular}[c]{@{}c@{}}0.01\\ (0.00)\end{tabular} & \begin{tabular}[c]{@{}c@{}}0.02\\ (0.01)\end{tabular} & \begin{tabular}[c]{@{}c@{}}0.02\\ (0.01)\end{tabular} & \begin{tabular}[c]{@{}c@{}}0.01\\ (0.00)\end{tabular} & \begin{tabular}[c]{@{}c@{}}0.01\\ (0.00)\end{tabular} & \begin{tabular}[c]{@{}c@{}}0.01\\ (0.01)\end{tabular} & \begin{tabular}[c]{@{}c@{}}0.01\\ (0.01)\end{tabular} \\ \hline
\end{tabular}
\end{table}

\begin{table}
\caption{Average (and optimal) values of the parameters $(c, k, z)$ in \eqref{eq:mce_cont_w}.}
\label{table:loss}
\centering
\begin{tabular}{lcccc}
                                                                                                                                                                                          \hline
\multicolumn{1}{c}{\begin{tabular}[c]{@{}c@{}}Parameter\end{tabular}} & C1                                                 & C2                                                 & C3                                                 & C4                                                 \\ \hline

\multicolumn{1}{c}{$c$}                                                       & \begin{tabular}[c]{@{}c@{}}1.45\\ (1.86)\end{tabular} & \begin{tabular}[c]{@{}c@{}}1.47\\ (1.86)\end{tabular} & \begin{tabular}[c]{@{}c@{}}0.80\\ (1.27)\end{tabular} & \begin{tabular}[c]{@{}c@{}}0.71\\ (0.80)\end{tabular} \\
\multicolumn{1}{c}{$k$}                                                       & \begin{tabular}[c]{@{}c@{}}2.30\\ (2.36)\end{tabular} & \begin{tabular}[c]{@{}c@{}}2.29\\ (2.36)\end{tabular} & \begin{tabular}[c]{@{}c@{}}0.26\\ (0.34)\end{tabular} & \begin{tabular}[c]{@{}c@{}}0.71\\ (0.75)\end{tabular} \\
\multicolumn{1}{c}{$z$}                                                       & \begin{tabular}[c]{@{}c@{}}2.43\\ (2.40)\end{tabular} & \begin{tabular}[c]{@{}c@{}}2.39\\ (2.40)\end{tabular} & \begin{tabular}[c]{@{}c@{}}-1.83\\ (-1.76)\end{tabular} & \begin{tabular}[c]{@{}c@{}}0.46\\ (0.46)\end{tabular} \\
 \hline
\end{tabular}
\end{table}

Figures~\ref{fig:compare} and \ref{fig:compare2} show the results of the Bayesian and Gibbs models for one simulation run in each of Examples B1--B2 and C1--C4, respectively.  The $95\%$ credible regions, as in \citet{Li.Ghosal.2015}, are highlighted in gray around the posterior means.  That is, let $u_i = \sup_\theta \{|\gamma_i(\theta) - \hat{\gamma}(\theta)|/s(\theta)\}$, where $\gamma_i(\theta)$ is the $i^{\text{th}}$ posterior boundary sample, $\hat{\gamma}(\theta)$ is the pointwise posterior mean and $s(\theta)$ the pointwise standard deviation of the $\gamma(\theta)$ samples.  If $\tau$ is the $95^{\text{th}}$ percentile of the $u_i$'s, then a $95\%$ uniform credible band is given by $\hat{\gamma}(\theta) \pm \tau \, s(\theta)$.  The results of cases B2 and C2 suggest that free-knot b-splines may do a better job of approximating non-smooth boundaries than the kernel basis functions used by \citet{Li.Ghosal.2015}.  In particular, the RJ-MCMC sampling method with its relocation moves allowed knots to move towards the corners of the triangle, thereby improving estimation of the boundary over b-splines with fixed knots.

\begin{figure}[t]
	\centering
		\includegraphics[width=.8\textwidth]{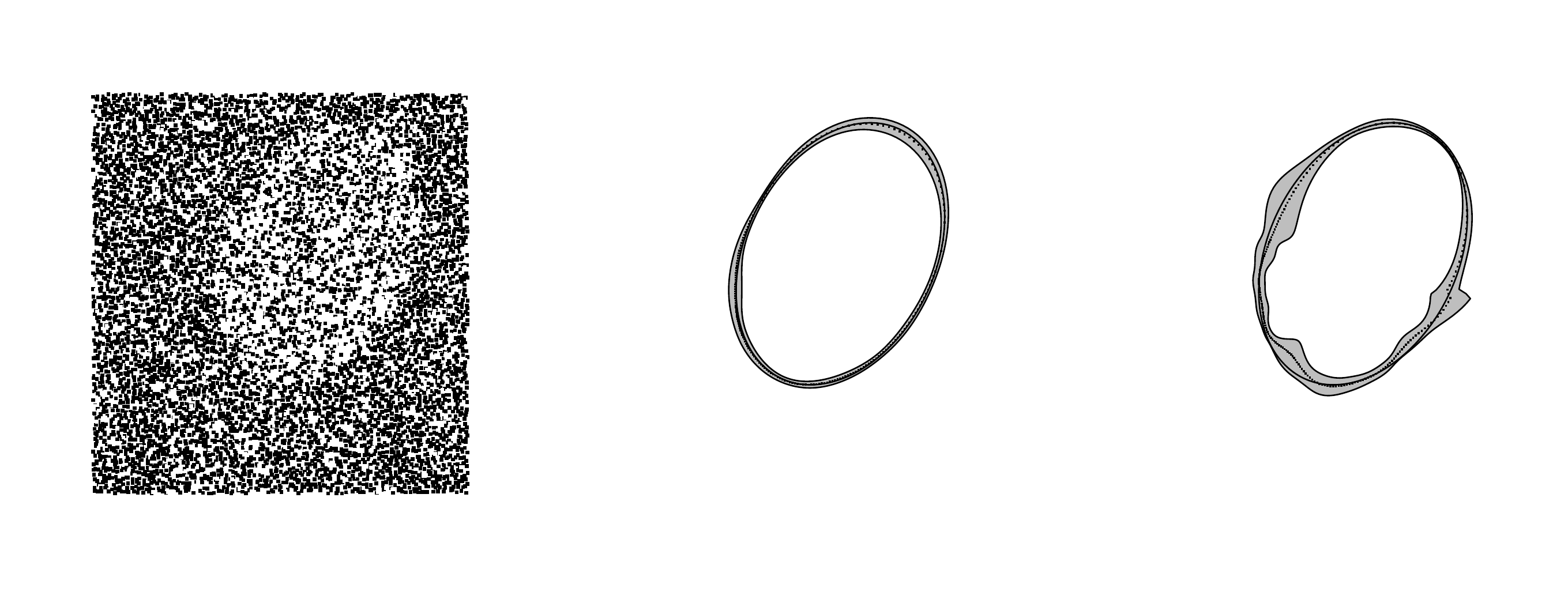}
		\includegraphics[width=.8\textwidth]{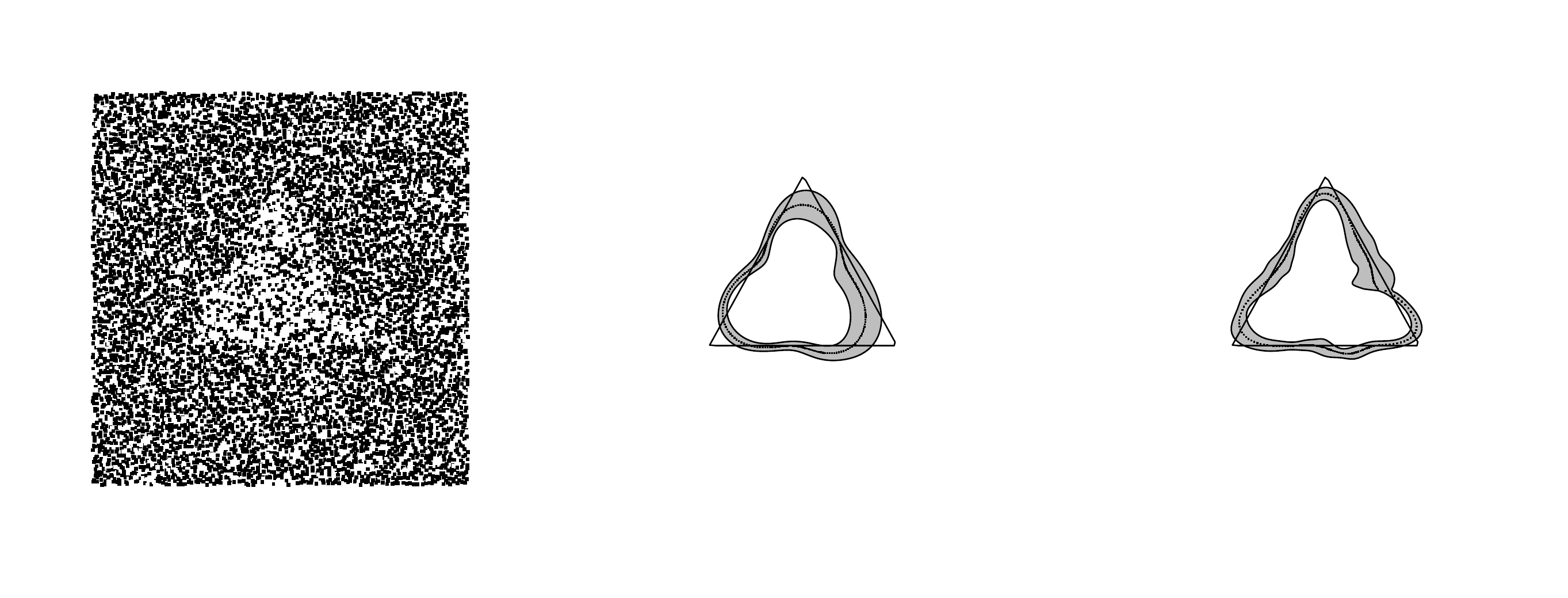}
\caption{From top, Examples B1--B2.  In each row, the observed image is on the left, the Bayesian posterior mean estimator \citep{Li.Ghosal.2015} is in the middle, and the Gibbs posterior mean estimator is on the right.  Solid lines show the true image boundary, dashed lines are the estimates, and gray regions are 95\% credible bands.}
\label{fig:compare}
\end{figure}

\begin{figure}
	\centering
		\includegraphics[width=.8\textwidth]{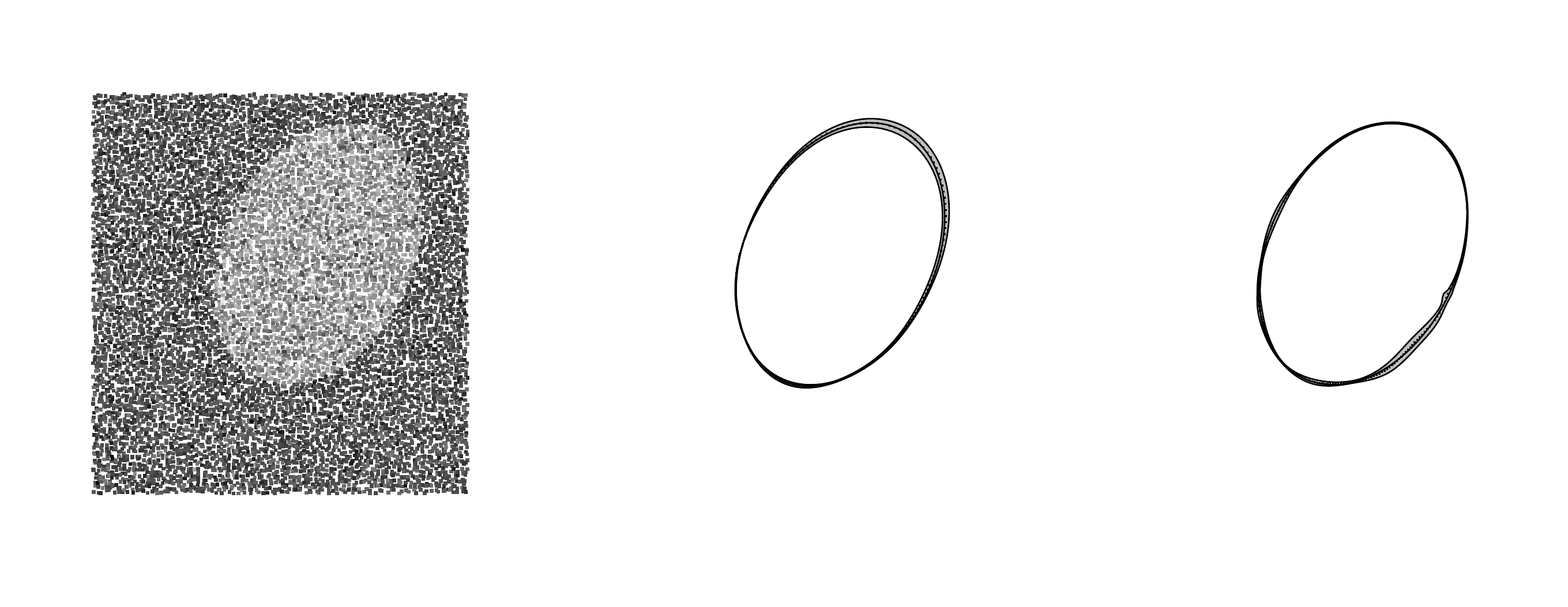}
		\includegraphics[width=.8\textwidth]{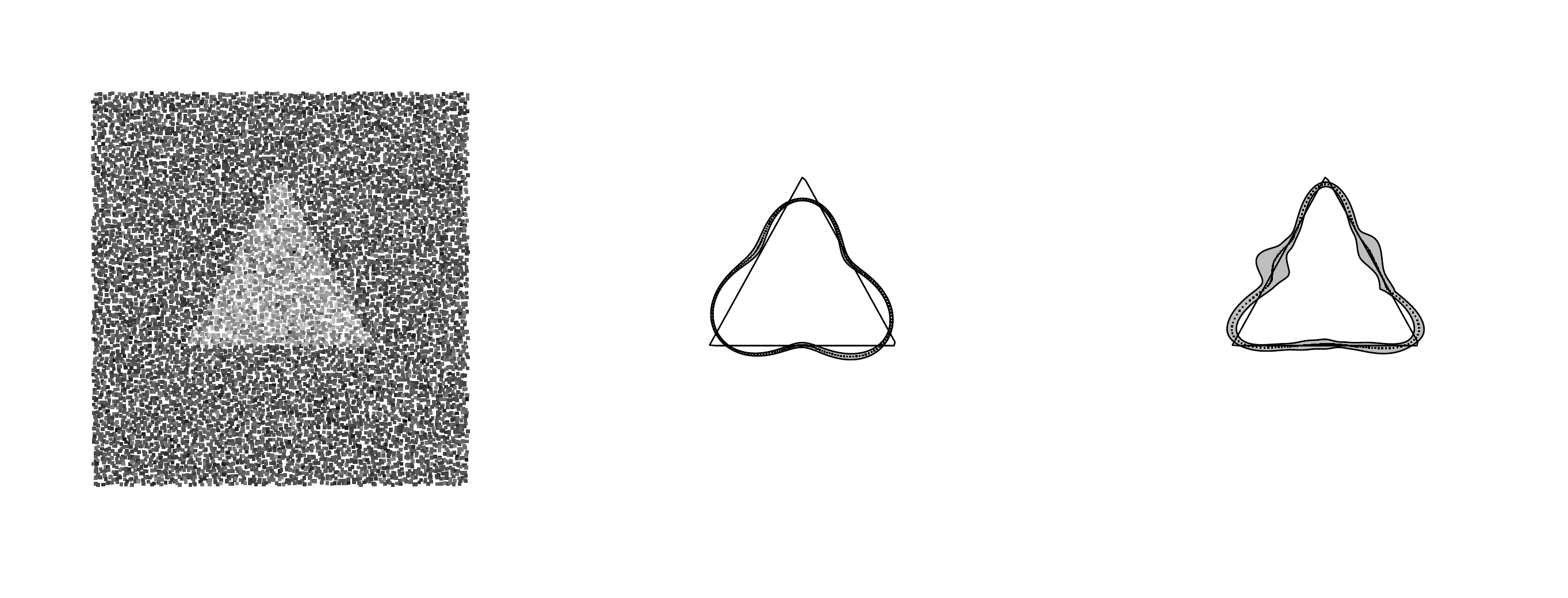}
		\includegraphics[width=.8\textwidth]{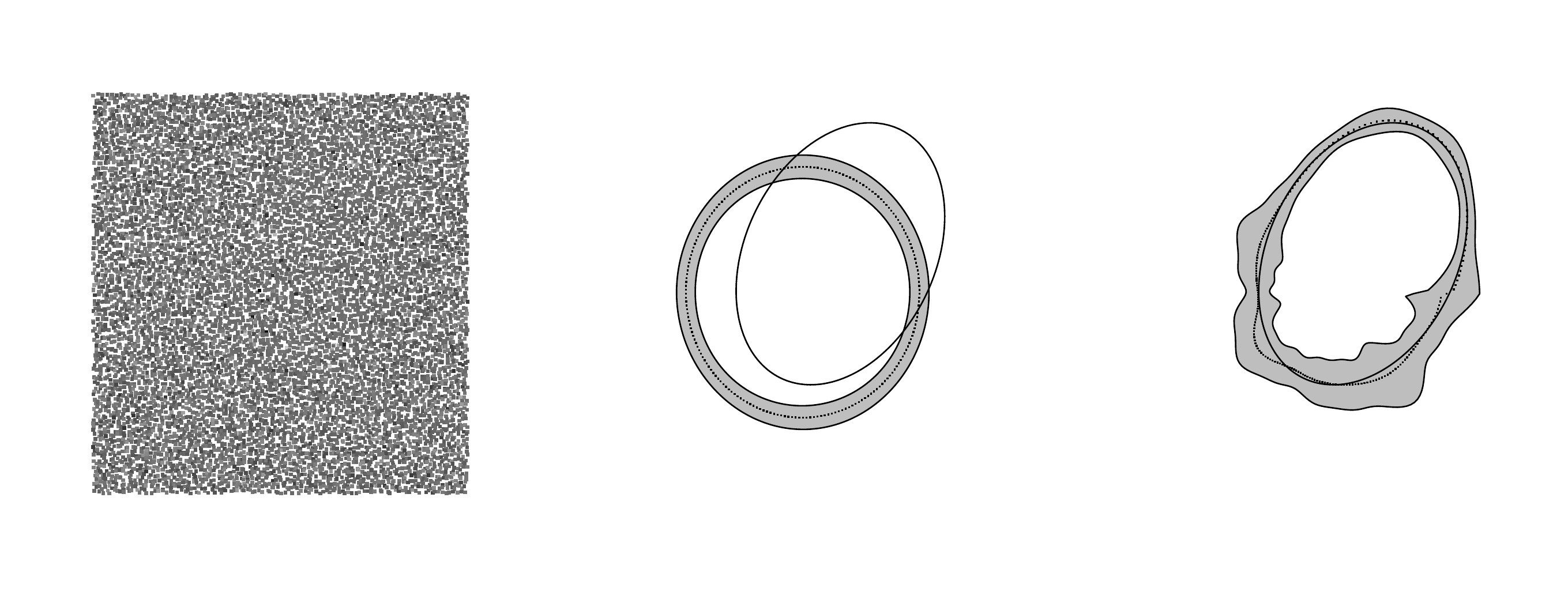}
		\includegraphics[width=.8\textwidth]{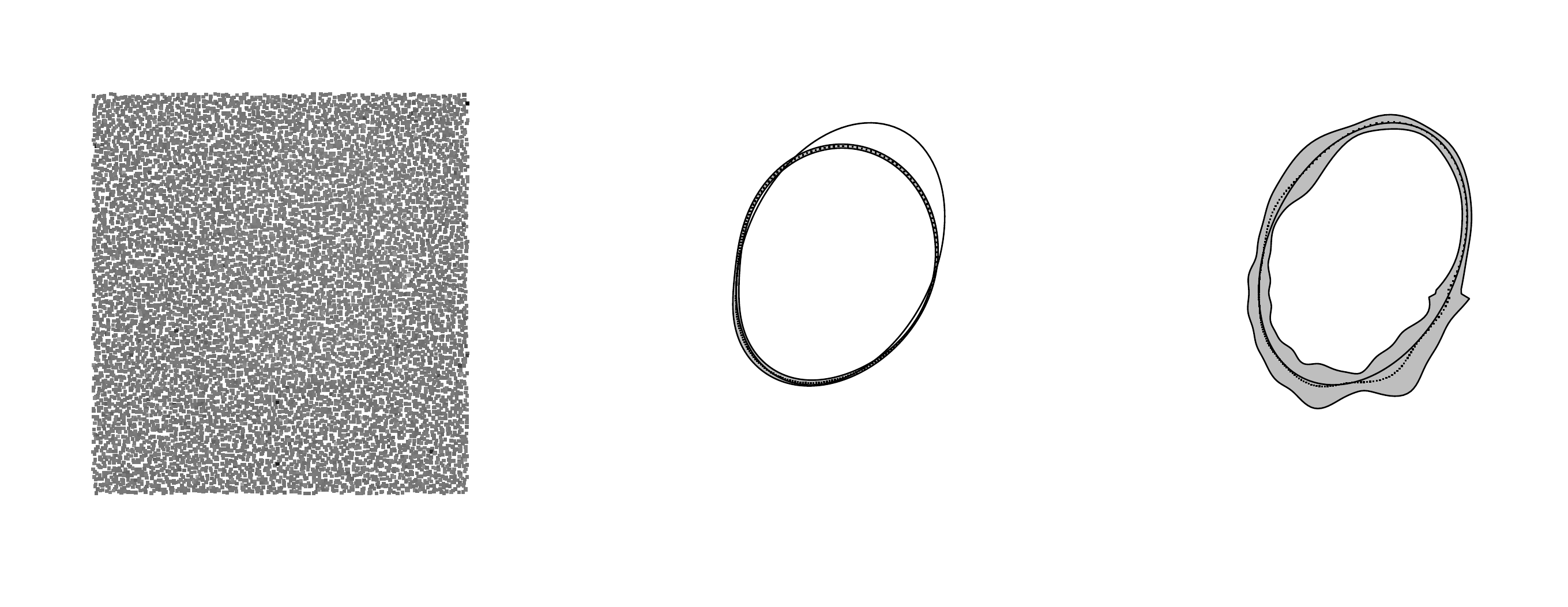}
\caption{Same as Figure~\ref{fig:compare}, but for Examples C1--C4.
}
\label{fig:compare2}
\end{figure}

\appendix

\section{Proofs}
\label{S:proofs}

\subsection{Proof of Proposition~1}
\label{p:prop1}

By the definition of the loss function in \eqref{eq:mce_w}, for a fixed $h$ and for any $\Gamma \subset \Omega$, we have 
\begin{align*}
\ell_\Gamma(X,Y) - \ell_{\Gamma^\star}(X,Y) & =h \, 1(Y=+1, X \in \Gamma^c) - h \, 1(Y=+1, X \in \Gamma^{\star c})\\
& \qquad +1(Y=-1, X \in \Gamma) - 1(Y=-1, X \in \Gamma^\star) \\
&=h \, 1(Y=+1, X \in \Gamma^\star \setminus \Gamma) - 1(Y=-1, X \in \Gamma^{\star} \setminus \Gamma)\\
& \qquad +1(Y=-1, X \in \Gamma \setminus \Gamma^\star) - h\,1(Y=+1, X \in \Gamma \setminus \Gamma^\star).
\end{align*}
Then the expectation of the loss difference above is 
\[P(X \in \Gamma^\star \setminus \Gamma) \, (h\fGammastar + \fGammastar - 1)+ P(X \in \Gamma \setminus \Gamma^\star) \, (1-\fGammastarc - h \fGammastarc), \]
where the probability statement is with respect to the density $g$ of $X$.  By Assumption~2, the density $g$ is bounded away from zero on $\Omega$, so the expectation of the loss difference is zero if and only if $\Gamma = \Gamma^\star$.  The expectation can also be lower bounded by 
\[ P(X \in \Gamma \sdiff \Gamma^\star) \, \min\{h\fGammastar + \fGammastar - 1, 1-\fGammastarc - h \fGammastarc \}. \]
Given the condition \eqref{eq:h.bound} in Proposition~1, both terms in the minimum are positive.  Therefore, $R(\Gamma) \geq R(\Gamma^\star)$ with equality if and only if $\Gamma = \Gamma^\star$, proving the claim.

\subsection{Proof of Proposition~2}
\label{p:prop2}

The proof here is very similar to that of Proposition~1.  By the definition of the loss function in \eqref{eq:mce_cont_w}, for any fixed $(c,k,z)$ and for any $\Gamma \subset \Omega$, we get 
\begin{align*}
\ell_\Gamma(X,Y) - \ell_{\Gamma^\star}(X,Y) & = k \, 1(Y\geq z, X \in \Gamma^c) - k \, 1(Y\geq z, X \in \Gamma^{\star c})\\
& \qquad +c \, 1(Y<z, X \in \Gamma) - c \, 1(Y<z, X \in \Gamma^\star)\\
&=k \, 1(Y\geq z, X \in \Gamma^\star \setminus \Gamma) - c \, 1(Y<z, X \in \Gamma^{\star} \setminus \Gamma)\\
& \qquad +c \, 1(Y<z, X \in \Gamma \setminus \Gamma^\star) - k \, 1(Y\geq z, X \in \Gamma \setminus \Gamma^\star).
\end{align*}
Then, the expectation of the loss difference above is given by
\[P(X \in \Gamma^\star \setminus \Gamma) \, \{k - k\FGammastar(z) - c\FGammastar(z)\} + P(X \in \Gamma \setminus \Gamma^\star) \, \{c\FGammastarc(z) - k + k\FGammastarc(z)\}, \]
where the probability statement is with respect to the density $g$ of $X$.  As in the proof of Proposition~1, this quantity is zero if and only if $\Gamma = \Gamma^\star$.  It can also be lower bounded by 
\[ P(X \in \Gamma \sdiff \Gamma^\star) \, \min\bigl\{k - k\FGammastar(z) - c\FGammastar(z), c\FGammastarc(z) - k + k\FGammastarc(z) \bigr\}. \]
Given the condition \eqref{eq:ckz.bound.1} in Proposition~2, both terms in the minimum are positive.  Therefore, $R(\Gamma) \geq R(\Gamma^\star)$ with equality if and only if $\Gamma = \Gamma^\star$, proving the claim.

\subsection{Preliminary results}
\label{p:thm}

Towards proving Theorem~\ref{thm:rate}, we need several lemmas.  The first draws a connection between the distance between defined by the Lebesgue measure of the symmetric difference and the sup-norm between the boundary functions.

\begin{lemma}
\label{lem:lebesgue.L1}
Suppose $\Gamma^\star$, with boundary $\gamma^\star = \partial \Gamma^\star$, satisfies Assumption~\ref{A:true}, in particular, $\underline{\gamma}^\star := \inf_{\theta \in [0,2\pi]} \gamma^\star(\theta) > 0$.  Take any $\Gamma \subset \Omega$, with $\gamma = \partial\Gamma$, such that $\lambda(\Gamma \sdiff \Gamma^\star) > \delta$ for some fixed $\delta>0$, and any $\widetilde\Gamma \subset \Omega$ such that $\tilde\gamma = \partial\widetilde\Gamma$ satisfies $\|\tilde\gamma - \gamma\|_\infty < \omega \delta$, where $\omega \in (0,1)$.  Then 
\[ \lambda(\widetilde\Gamma \sdiff \Gamma^\star) > \frac{4\delta}{\underline{\gamma}^\star} \Bigl( \frac{1}{\diam(\Omega)} - \pi \omega \Bigr), \]
where $\diam(\Omega) = \sup_{x,x' \in \Omega} \|x-x'\|$ is the diameter of $\Omega$.  So, if $\omega < \{\pi \diam(\Omega)\}^{-1}$, then the lower bound is a positive multiple of $\delta$.
\end{lemma}

\begin{proof}
The first step is to connect the symmetric difference-based distance to the $L_1$ distance between boundary functions.  A simple conversion to polar coordinates gives 
\begin{align*}
\lambda(\Gamma \sdiff \Gamma^\star) & = \int_{\Gamma \sdiff \Gamma^\star} \, d\lambda \\
& = \int_0^{2\pi} \int_{\gamma(\theta) \wedge \gamma^\star(\theta)}^{\gamma(\theta) \vee \gamma^\star(\theta)} r \,dr \, d\theta \\
& = \frac12 \int_0^{2\pi} \{\gamma(\theta) \wedge \gamma(\theta^\star)\}^2 - \{\gamma(\theta) \vee \gamma(\theta^\star)\}^2 \, d\theta \\ 
& = \frac12 \int_0^{2\pi} |\gamma(\theta) - \gamma^\star(\theta)|\,|\gamma(\theta) + \gamma^\star(\theta)| \,d\theta.
\end{align*}
If we let $\underline{\gamma}^\star = \inf_\theta \gamma^\star(\theta)$, then it is easy to verify that 
\[ \underline{\gamma}^\star \leq |\gamma(\theta) + \gamma^\star(\theta)| \leq \diam(\Omega), \quad \forall \; \theta \in [0,2\pi]. \]
Therefore, 
\begin{equation}
\label{eq:lebesgue.L1}
\tfrac12 \underline{\gamma}^\star \|\gamma - \gamma^\star\|_1 \leq \lambda(\Gamma \sdiff \Gamma^\star) \leq \tfrac12 \diam(\Omega) \|\gamma - \gamma^\star\|_1. 
\end{equation}
Next, if $\lambda(\Gamma \sdiff \Gamma^\star) > \delta$, which is positive by Assumption~\ref{A:true}, then it follows from the right-most inequality in \eqref{eq:lebesgue.L1} that $\diam(\Omega) \|\gamma-\gamma^\star\|_1 > 2\delta$ and, by the triangle inequality, 
\[ \diam(\Omega) \{\|\gamma - \tilde\gamma\|_1 + \|\tilde\gamma - \gamma^\star\|_1\} > 2\delta. \]
We have $\|\gamma-\tilde\gamma\|_1 \leq 2\pi\|\gamma - \tilde\gamma\|_\infty$ which, by assumption, is less than $2\pi \omega \delta$.  Consequently, 
\[ \diam(\Omega) \{2\pi\omega \delta + \|\tilde\gamma - \gamma^\star\|_1\} > 2\delta \]
and, hence, 
\[ \|\tilde\gamma - \gamma^\star\|_1 > \frac{2\delta}{\diam(\Omega)} - 2\pi\omega\delta. \]
By the left-most inequality in \eqref{eq:lebesgue.L1}, we get 
\[ \lambda(\widetilde\Gamma \sdiff \Gamma^\star) > \frac{4\delta}{\underline{\gamma}^\star \diam(\Omega)} - \frac{4\pi\omega \delta}{\underline{\gamma}^\star} = \frac{4\delta}{\underline{\gamma}^\star} \Bigl( \frac{1}{\diam(\Omega)} - \pi \omega \Bigr), \]
which is the desired bound.  It follows immediately that the lower bound is a positive multiple of $\delta$ if $\omega < \{\pi \diam(\Omega)\}^{-1}$.  
\end{proof}

The next lemma shows that we can control the expectation of the integrand in the Gibbs posterior under the condition \eqref{eq:ckz.bound.2} on the tuning parameters $(c,k,z)$ in the loss function definition.  

\begin{lemma}
\label{lem:num}
If \eqref{eq:ckz.bound.2} holds, then $\PGammastar e^{-(\ell_\Gamma -\ell_{\Gamma^\star})} < 1-\rho\lambda(\Gamma^\star \sdiff \Gamma)$ for a constant $\rho \in (0,1)$.
\end{lemma}

\begin{proof}
From the proof of Proposition~2, we have
\begin{align*}
\ell_\Gamma(x,y) - \ell_{\Gamma^\star}(x,y) & = k \, 1(y \geq z, x \in \Gamma^\star \setminus \Gamma) - c \, 1(y<z, x \in \Gamma^{\star} \setminus \Gamma) \\
& \qquad +c \, 1(y<z, x \in \Gamma \setminus \Gamma^\star) - k \, (y\geq z, x \in \Gamma \setminus \Gamma^\star).
\end{align*}
The key observation is that, if $x \not\in \Gamma \sdiff \Gamma^\star$, then the loss difference is 0 and, therefore, the exponential of the loss difference is 1.  Taking expectation with respect $\PGammastar$, we get 
\begin{align*}
\PGammastar e^{-(\ell_\Gamma -\ell_{\Gamma^\star})} & = P_g(X \not\in \Gamma^\star \sdiff \Gamma) \\
& \qquad + \{e^{-k}(1-\FGammastar(z))+e^c\FGammastar(z)\} P_g(X \in \Gamma^\star \setminus \Gamma) \\
& \qquad + \{e^{-c}\FGammastarc(z)+e^k - e^k\FGammastarc(z)\} P_g(X \in \Gamma \setminus \Gamma^\star).
\end{align*}
From \eqref{eq:ckz.bound.2}, we have 
\[ \kappa := \max\{ e^{-k}(1-\FGammastar(z))+e^c\FGammastar(z), e^{-c}\FGammastarc(z)+e^k - e^k\FGammastarc(z) \} < 1, \]
so that
\begin{align*}
\PGammastar e^{-(\ell_\Gamma - \ell_{\Gamma^\star})} & \leq 1 - P_g(X \in \Gamma \sdiff \Gamma^\star) + \kappa P_g(X \in \Gamma \sdiff \Gamma^\star) \\
& = 1 - (1-\kappa) P_g(X \in \Gamma \sdiff \Gamma^\star). 
\end{align*}
Then the claim follows, with $\rho = (1-\kappa) \underline{g} < 1$, since $P_g(X \in \Gamma \sdiff \Gamma^\star) \geq \underline{g} \lambda(\Gamma \sdiff \Gamma^\star)$.    
\end{proof}

The next lemma yields a necessary lower bound on the denominator of the Gibbs posterior distribution.  The neighborhoods $G_n$ are simpler than those in Lemma~1 of \citet{shen.wasserman.2001}, because the variance of our loss difference is automatically of the same order as its expectation, but the proof is otherwise very similar to theirs, so we omit the details.     

\begin{lemma}
\label{lem:den}
Let $t_n$ be a sequence of positive numbers such that $nt_n \rightarrow 0$ and set $G_n = \{\Gamma: R(\Gamma)-R(\Gamma^\star)\leq Ct_n\}$ for some $C>0$.  Then $\int e^{-[R_n(\Gamma)-R_n(\Gamma^\star)]}\,d\Pi \gtrsim \Pi(G_n)\exp(-2nt_n)$ with $P_{\Gamma^\star}$-probability converging to $1$ as $n\rightarrow \infty$.
\end{lemma}

Our last lemma is a summary of various results derived in \citet{Li.Ghosal.2015} towards proving their Theorem~3.3.  This helps us to fill in the details for the lower bound in Lemma~\ref{lem:den} and to identify a sieve---a sequence of subsets of the parameter space---with high prior probability but relatively low complexity.  

\begin{lemma}
\label{lem:li.ghosal}
Let $\eps_n$ be as in Theorem~\ref{thm:rate} and let $D_n = (\frac{n}{\log n})^{\frac{1}{\alpha+1}}$.  Then, $\|\gamma^\star - \hat\gamma_{D_n, \beta^\star}\|_\infty\leq C\eps_n$ for some $C>0$, $\beta^\star = \beta_{D_n}^\star \in  (\RR^+)^{D_n}$ from \eqref{eq:bsplineapprox}.  
\label{lem:prior}
\begin{enumerate}
\item Define the neighborhood $B_n^\star = \{(\beta,d): \beta \in \mathbb{R}^d, d=D_n, \|\gamma^\star - \hat\gamma_{d, \beta}\|_\infty\leq C\eps_n\}$.  Then $\Pi(B_n^\star) \gtrsim  \exp(-an\eps_n)$ for some $a > 0$ depending on $C$.
\item Define the sieve $\Sigma_n = \{\gamma: \gamma = \hat\gamma_{d,\beta}, \beta \in \RR^d, d\leq D_n, \|\beta\|_\infty\leq \sqrt{n/K_0}\}$.  Then $\Pi(\Sigma_n^c) \lesssim \exp(-Kn\eps_n)$ for some $K,\, K_0>0$.  
\item The bracketing number of $\Sigma_n$ satisfies $\log N(\eps_n, \Sigma_n,\|\cdot\|_\infty)\lesssim n\eps_n$.
\end{enumerate}
\end{lemma}

\subsection{Proof of Theorem~\ref{thm:rate}}

Define the set
\begin{equation}
\label{eq:proof1}
A_n = \{\Gamma: \lambda(\Gamma^\star \sdiff \Gamma) > M\eps_n\}.
\end{equation}
For the sieve $\Sigma_n$ in Lemma~\ref{lem:li.ghosal}, we have $\Pi_n(A_n)\leq\Pi_n(\Sigma_n^c) + \Pi_n(A_n \cap \Sigma_n)$.  We want to show that both terms in the upper bound vanish, in $L_1(\PGammastar)$. 

It helps to start with a lower bound on $I_n = \int e^{-n\{R_n(\Gamma) - R_n(\Gamma^\star)\}} \, \Pi(d\Gamma)$, the denominator in both of the terms discussed above.  First, write 
\[I_n \geq \int_{G_n}e^{-n\{R_n(\Gamma) - R_n(\Gamma^\star)\}} \, \Pi(d\Gamma)\] where $G_n$ is defined in Lemma~\ref{lem:den} with $t_n = \eps_n$ and $C>0$ to be determined.  From Proposition~2 and Lemma~\ref{lem:lebesgue.L1}
\begin{align*}
R(\Gamma) - R(\Gamma^\star)&\leq P_g(X \in \Gamma \sdiff \Gamma^\star) \min\{ k - k\FGammastar(z) - c\FGammastarc(z), c\FGammastarc(z) + k\FGammastarc(z) \} \\
& \leq \tfrac12 \, V \, \overline{g} \, \diam(\Omega) \, \|\gamma - \gamma^\star\|_1 
\end{align*}
where $V=V_{c,k,z}$ is the $\min\{\cdots\}$ term in the above display.  Let $B_\infty(\gamma^\star; r)$ denote the set of regions $\Gamma$ with boundary functions $\gamma = \partial \Gamma$ that satisfy $\|\gamma - \gamma^\star\|_\infty \leq r$.  If $\Gamma \in B_\infty(\gamma^\star; C_0 \eps_n)$, then we have 
\[ \|\gamma - \gamma^\star\|_1 \leq 2\pi \overline{g} C_0 \eps_n \]
and, therefore, $R(\Gamma) - R(\Gamma^\star) \leq C \eps_n$, where 
$C = C_0\pi V\overline{g}^2 \diam(\Omega)$.  From Lemma~\ref{lem:den}, we have $I_n \gtrsim \Pi(G_n) e^{-2C\eps_n}$, with $\PGammastar$-probability converging to $1$.  Since $G_n \supseteq B_\infty(\gamma^\star; C_0 \eps_n)$, it follows from Lemma~\ref{lem:li.ghosal}, part~1, 
\[ I_n \gtrsim \Pi\{B_\infty(\gamma^\star; C_0 \eps_n)\} e^{-2C \eps_n} \gtrsim e^{-C_1 n \eps_n}, \]
with $\PGammastar$ probability converging to one, and where $C_1 > 0$ is a constant depending on $C_0$ and $C$.   

Now we are ready to bound $\Pi_n(\Sigma_n^c)$.  Write this quantity as 
\[\Pi_n(\Sigma_n^c) =\frac{N_n(\Sigma_n^c)}{I_n} = \frac{1}{I_n} \int_{\Sigma_n^c} e^{-n\{R_n(\Gamma) - R_n(\Gamma^\star\}} \, \Pi(d\Gamma).\]
It will suffice to bound the expectation of $N_n(\Sigma_n^c)$.  By Tonelli's theorem, independence, and Lemma~\ref{lem:num}, we have 
\begin{align*}
\PGammastar N_n(\Sigma_n^c) & = \int_{\Sigma_n^c} \PGammastar e^{-n\{R_n(\Gamma) - R_n(\Gamma^\star)\}} \,\Pi(d\Gamma) \\
& =\int_{\Sigma_n^c} \bigl\{ \PGammastar e^{-(\ell_\Gamma - \ell_{\Gamma^\star})} \}^n \,\Pi(d\Gamma) \\
& \leq \int_{\Sigma_n^c} \bigl\{ 1 - \rho \lambda(\Gamma \sdiff \Gamma^\star) \}^n \, \Pi(d\Gamma) \\
& \leq \Pi(\Sigma_n^c).
\end{align*}
By Lemma~\ref{lem:li.ghosal}, part~2, we have that $\Pi(\Sigma_n^c) \leq e^{-K n \eps_n}$.  

Next, we bound $\Pi_n(A_n \cap \Sigma_n)$.  Again, it will suffice to bound the expectation of $N_n(A_n \cap \Sigma_n)$.  Choose a covering $A_n \cap \Sigma_n$ by sup-norm balls $B_j=B_\infty(\gamma_j;\omega M_n\eps_n)$, $j=1,\ldots,J_n$, with centers $\gamma_j=\partial\Gamma_j$ in $A_n$ and radii $\omega M_n \eps_n$, where $\omega < \{\pi \diam(\Omega)\}^{-1}$ is as in Lemma~\ref{lem:lebesgue.L1}.  Also, from Lemma~\ref{lem:li.ghosal}, part~3, we have that $J_n$ is bounded by $e^{K_1n\eps_n}$ for some constant $K_1 > 0$.  For this covering, we immediately get 
\[ \PGammastar N_n(A_n \cap \Sigma_n) \leq \sum_{j=1}^{J_n} \PGammastar N_n(B_j). \]
For each $j$, using Tonelli, independence, and Lemma~\ref{lem:num} again, we get 
\[ \PGammastar N_n(B_j) = \int_{B_j} \{ \PGammastar e^{-(\ell_\Gamma - \ell_{\Gamma^\star})} \}^n \, \Pi(d\Gamma) \leq \int_{B_j} e^{-n \rho \lambda(\Gamma \sdiff \Gamma^\star)} \, \Pi(d\Gamma). \]
By Lemma~\ref{lem:lebesgue.L1}, for $\Gamma$ in $B_j$, since the center $\gamma_j$ is in $A_n$, it follows that $\lambda(\Gamma \sdiff \Gamma^\star)$ is lower bounded by $\eta M\eps_n$, where $\eta=\eta(\Gamma^\star)$ is given by 
\[ \eta = \frac{4}{\underline{\gamma}^\star} \Bigl( \frac{1}{\diam(\Omega)} - \pi \omega \Bigr) > 0.  \]
Therefore, from the bound on $J_n$, 
\[ \PGammastar N_n(A_n \cap \Sigma_n) \leq \sum_{j=1}^{J_n} \PGammastar N_n(B_j) \leq e^{- n \eta M \eps_n} \, J_n \leq e^{-(\eta M - K_1) n \eps_n}.  \]

Finally, we have 
\begin{align*}
\Pi_n(A_n) & \leq \Pi_n(A_n \cap \Sigma_n) + \Pi_n(\Sigma_n^c) \\
& = \frac{N_n(A_n \cap \Sigma_n)}{I_n} + \frac{N_n(\Sigma_n^c)}{I_n} \\
& \leq \frac{N_n(A_n \cap \Sigma_n)}{e^{-C_1n\eps_n}} \, 1(I_n > e^{-C_1n\eps_n}) + \frac{N_n(\Sigma_n^c)}{I_n} \, 1(I_n \leq e^{-C_1n\eps_n}) \\
& \leq \frac{N_n(A_n \cap \Sigma_n)}{e^{-C_1n\eps_n}} + 1(I_n \leq e^{-C_1n\eps_n}).
\end{align*}
Taking $\PGammastar$-expectation and plugging in the bounds derived above, we get 
\[ \PGammastar \Pi_n(A_n) \leq e^{-(\eta M - K_1-C_1)n\eps_n}. \]
If $M > (K_1 + C_1)/\eta$, then 
the upper bound vanishes, completing the proof.   


\bibliographystyle{apalike}
\bibliography{mybib_i}

\end{document}